\documentclass[letterpaper, 10 pt, conference]{ieeeconf}  

\IEEEoverridecommandlockouts                              

\usepackage{cite}
\usepackage{amsmath,amssymb,amsfonts,amsthm}

\usepackage{etoolbox}

\usepackage{algpseudocode}
\usepackage[linesnumbered, ruled, vlined, noend]{algorithm2e}
\SetAlFnt{\footnotesize}
\SetAlCapFnt{\footnotesize}
\SetAlCapNameFnt{\footnotesize}
\SetArgSty{textrm}
\SetInd{0.1em}{0.5em}
\SetCommentSty{emph}
\SetAlgoSkip{}
\makeatletter
\patchcmd{\@algocf@start}
  {-1.5em}
  {0pt}
  {}{}
  
\usepackage{graphicx}
\usepackage[export]{adjustbox}
\usepackage{textcomp}
\usepackage[table]{xcolor}
\usepackage{tikz}
\usetikzlibrary{arrows.meta,calc,positioning,fit,bending,backgrounds}
\usepackage{booktabs}
\usepackage{multirow}

\newtheorem{theorem}{Theorem}

\newtheorem{lemma}[theorem]{Lemma}
\newtheorem{proposition}[theorem]{Proposition}

\def\BibTeX{{\rm B\kern-.05em{\sc i\kern-.025em b}\kern-.08em
    T\kern-.1667em\lower.7ex\hbox{E}\kern-.125emX}}

\allowdisplaybreaks
\usepackage{setspace}

\newcommand{\SoC}{\mathrm{SoC}}
    
\begin{document}

\title{\LARGE \bf
Goal Staying Makes Sum-of-Costs\\ Anonymous Multi-Agent Path Finding NP-Hard
}

\author{Hang Ma, \textit{Simon Fraser University}, \texttt{hangma@sfu.ca}}

\maketitle
\thispagestyle{empty}
\pagestyle{empty}

\begin{abstract} Anonymous Multi-Agent Path Finding (AMAPF) admits polynomial-time network-flow algorithms for several objectives, including makespan, total distance, and sum-of-costs (SoC) when agents disappear upon reaching goals. We show that standard goal-staying AMAPF is fundamentally different. We first formulate SoC minimization by augmenting the standard time-expanded flow model with goal-settlement constraints and show that the resulting linear programming relaxation is non-integral. We then prove that minimizing SoC in goal-staying AMAPF is NP-hard via a reduction from 3-SAT. Together with the polynomial-time result for the disappearing variant, this establishes a sharp complexity boundary determined by whether completed agents remain at their goals. \end{abstract}

\section{Introduction}
\label{sec:intro}

Multi-Agent Path Finding (MAPF)~\cite{SternEtAl2019}, also called graph-based multi-robot path planning~\cite{ma2022graph}, aims to compute collision-free paths
for multiple agents from their start vertices to their goal vertices
on a graph.
In standard MAPF, each agent has a preassigned goal.
In \emph{Anonymous MAPF} (AMAPF), also called unlabeled,
goal-invariant, or permutation-invariant MAPF, the agents are
interchangeable and any agent can be assigned to any goal.
Equivalently, AMAPF is the one-team special case of Target Assignment
and Path Finding (TAPF)~\cite{MaKoenig2016}.

Full interchangeability leads to a striking reduction in computational
complexity.
While optimizing standard MAPF is NP-hard for common time-based objectives such
as makespan and sum-of-costs (SoC), AMAPF with the makespan objective
can be solved in polynomial time using single-commodity maximum flow
on a time-expanded network~\cite{YuLaValle2013Flow,AliYakovlev2024}. Here, SoC is the sum of completion times of all agents and thus counts both moves and waits before completion.
Minimum total traveled distance is also polynomial by assigning costs to the corresponding flow network~\cite{YuLaValle2013Flow}. This tractability has made network flow the natural algorithmic foundation for optimally solving several important AMAPF variants.

The situation for SoC is surprisingly different.
A standard modeling choice in MAPF is what happens after an agent
completes its path at a goal: it may either \emph{stay at the goal} and continue
participating in collisions, or \emph{disappear at the goal} and cease
to occupy the graph~\cite{SternEtAl2019}.
For disappearing AMAPF, SoC can be minimized in polynomial time using
minimum-cost maximum flow: a unit of flow simply leaves the
time-expanded network when it reaches a goal
\cite{YuLaValle2013Flow,AliYakovlev2026}.
Standard AMAPF instead uses goal-staying semantics.
For this setting, existing optimal approaches have exponential
worst-case complexity, including CBS with task assignment
and its variants~\cite{HonigEtAl2018,tang2023solving}, and no polynomial-time algorithm is
known~\cite{AliYakovlev2024,AliYakovlev2026}.
This leaves a long-standing complexity gap.
\emph{Why can flow efficiently optimize AMAPF makespan, total distance, and
disappearing-agent SoC, but not standard goal-staying SoC?
More fundamentally, is the latter problem polynomial-time
solvable at all?}

In this paper, we answer this question by showing that
\boxed{\emph{\textbf{goal staying makes optimal AMAPF-SoC NP-hard}}}.
The key distinction is that reaching a goal does not merely terminate
a unit of flow.
Once an agent completes its path, it permanently reserves its goal
vertex, coupling its completion decision with the motion of all agents
at future timesteps.

Our development starts from this observation.
We first augment the standard single-commodity time-expanded-flow
formulation with \emph{goal-settlement} variables that explicitly
model permanent goal occupancy.
The resulting formulation reveals why the usual network-flow
integrality argument no longer applies: its natural LP relaxation is
not integral.
This observation also provides intuition for our NP-hardness
construction, which uses permanent goal occupancy to encode mutually
exclusive choices and enforce Boolean consistency.

Our main contributions are:
\begin{enumerate}
    \item we identify permanent goal occupancy as the key constraint that
distinguishes standard AMAPF-SoC from polynomial flow-based AMAPF
variants and formulate it explicitly in a compact integer program;
    \item we show that the natural LP relaxation of this formulation is
non-integral,including an explicit instance with an integrality gap of $7/6$; and
    \item most importantly, we prove that minimizing SoC for standard goal-staying
AMAPF is NP-hard. Together with the polynomial-time result for disappearing AMAPF-SoC~\cite{YuLaValle2013Flow,AliYakovlev2026}, our result establishes a sharp
complexity boundary based solely on the behavior of agents after
reaching their goals.
\end{enumerate}

\begin{table*}[t]
\centering
\scriptsize
\caption{
Representative complexity results for graph-based MAPF and its anonymous and colored variants. 
}
\label{tab:mapf-complexity}
\setlength{\tabcolsep}{2.6pt}
\renewcommand{\arraystretch}{1.13}
\begin{tabular}{p{3.09cm}p{1.83cm}p{3.44cm}p{1.3cm}p{7.2cm}}
\hline
\textbf{Problem} &
\textbf{Graph} &
\textbf{Objective / Question} &
\textbf{Complexity} &
\textbf{Reference / Note} \\
\hline

Labeled MAPF &
general &
feasibility &
P &
pebble-motion \cite{KornhauserEtAl1984},\cite{YuRus2015}
\\

Labeled MAPF &
general &
makespan &
NP-hard &
\cite{Surynek2010}, $(4/3-\epsilon)$-inapprox. \cite{MaEtAl2016}
\\

Labeled MAPF &
general &
SoC, total/max distance &
NP-hard &
\cite{YuLaValle2016}
\\

Labeled MAPF &
directed &
feasibility &
NP-complete &
\cite{Nebel2024}, P for strongly connected directed \cite{ArdizzoniEtAl2025}
\\
\hline

Anonymous MAPF &
general &
makespan, total distance &
P &
\cite{YuLaValle2013Flow}
\\

Anonymous MAPF (disappear) &
general &
SoC &
P &
\cite{AliYakovlev2026}
\\

\rowcolor{gray!25}
{Anonymous MAPF} &
{general} &
{SoC} &
\textbf{NP-hard} &
\textbf{this work}
\\
\hline

2-colored MAPF / TAPF &
general &
makespan &
NP-hard &
\cite{YuLaValle2016}, $(4/3-\epsilon)$-inapprox. \cite{MaEtAl2016}, $(3/2-\epsilon)$-inapprox. also directed \cite{tan2025inapproximability}
\\

2-colored MAPF / TAPF &
general &
SoC &
NP-hard &
\cite{YuLaValle2016}, even when edge conflicts allowed \cite{MaEtAl2016}
\\
\hline

Labeled MAPF &
planar &
makespan, SoC,
total/max distance &
NP-hard &
\cite{Yu2016Planar}
\\

2-colored MAPF / TAPF &
planar &
makespan, SoC &
NP-hard &
\cite{Yu2016Planar}
\\
\hline

Labeled MAPF &
2D grid with holes &
makespan, SoC &
NP-hard &
\cite{BanfiEtAl2017}
\\

2-colored MAPF / TAPF &
2D grid with holes &
makespan, SoC &
NP-hard &
three directions suffice~\cite{Geft2023}
\\
\hline

\end{tabular}
\end{table*}

\section{Related Work: Complexity of MAPF Variants}
\label{sec:related}

Graph-based multi-robot path planning and MAPF refer to essentially the same discrete coordination problem.
The variant with interchangeable agents has also been called
\emph{anonymous}, \emph{unlabeled}, \emph{goal-invariant}, or
\emph{permutation-invariant} MAPF.  More generally, $k$-colored MAPF,
or TAPF, partitions the agents
into $k$ interchangeable teams~\cite{MaKoenig2016}.  Thus, one team
corresponds to AMAPF, while one agent per team corresponds to standard
labeled MAPF.

Feasibility of classical pebble-motion problems, including feasible
labeled MAPF on undirected graphs, is polynomial
\cite{KornhauserEtAl1984,YuRus2015}, whereas optimizing labeled MAPF is
NP-hard.  In particular, makespan minimization is NP-hard
\cite{Surynek2010}, and makespan, SoC, and distance-based objectives are
NP-hard on general graphs~\cite{YuLaValle2016}.
Full interchangeability changes this picture substantially:
AMAPF makespan and total traveled distance can be optimized in
polynomial time using single-commodity network-flow formulations
\cite{YuLaValle2013Flow,AliYakovlev2024}. Scalable complete but suboptimal AMAPF algorithms have also been developed~\cite{OkumuraDefago2023}.

For SoC, the behavior of agents after reaching goals becomes critical.
MAPF models commonly distinguish between agents that stay at their
goals and agents that disappear upon arrival~\cite{SternEtAl2019}.
For disappearing AMAPF, SoC is polynomially solvable by minimum-cost
maximum flow~\cite{YuLaValle2013Flow,AliYakovlev2026}.
For standard goal-staying AMAPF, however, existing optimal algorithms
have exponential worst-case complexity~\cite{HonigEtAl2018,
AliYakovlev2024,AliYakovlev2026}.
This work closes the corresponding complexity gap by proving
goal-staying optimal AMAPF-SoC NP-hard.

Partial interchangeability is already sufficient for hardness.
With only two teams, MAPF is NP-hard for both makespan and SoC on
general graphs~\cite{MaEtAl2016,ma2020target}.
For makespan, \cite{MaEtAl2016} established $(4/3-\epsilon)$
inapproximability, and \cite{tan2025inapproximability} strengthened
this to $(3/2-\epsilon)$ even for instances of optimal makespan two.
Thus, for makespan, the tractability boundary between one and two teams
is already sharp.

Hardness also persists on restricted graph classes.
On planar graphs, labeled MAPF is NP-hard for makespan, SoC, total
distance, and maximum individual distance, while the two time-based
objectives remain NP-hard with only two teams~\cite{Yu2016Planar}.
On 2D grid graphs with holes, labeled makespan and SoC are NP-hard
\cite{BanfiEtAl2017}; \cite{Geft2023} further established hardness for two-team
MAPF even when agents move in only three cardinal directions.
On directed graphs, even MAPF feasibility becomes NP-complete in
general~\cite{Nebel2024}, although it is polynomial on strongly
connected digraphs~\cite{Nebel2024,ArdizzoniEtAl2025}; optimal directed
MAPF remains hard under strong restrictions~\cite{TanBercher2023}.

Table~\ref{tab:mapf-complexity} summarizes the results most relevant
to the complexity boundary studied in this paper. It highlights two complementary boundaries.
Partial interchangeability already causes hardness: moving from one
team to two makes the standard time objectives NP-hard.
Our result identifies a different boundary \emph{within} fully
anonymous MAPF: disappearing-agent SoC is polynomial, whereas
goal-staying SoC is NP-hard.

Interchangeability has also been studied in continuous-space multi-robot motion planning. Early work considered path planning directly in a permutation-invariant multi-robot formation space \cite{KloderHutchinson2006} and jointly optimizing assignment and collision-free trajectories for unlabeled robots~\cite{TurpinEtAl2014}. The $k$-color formulation considers multiple interchangeable robot groups \cite{SoloveyHalperin2014}. For fully unlabeled unit discs, polynomial-time algorithms are known under geometric separation assumptions: \cite{AdlerEtAl2015} studied feasibility in simple polygons, while \cite{SoloveyEtAl2015} provided near-optimal total path-length guarantees in polygonal environments. Without such restrictions, even unlabeled motion planning is PSPACE-hard for unit-square robots among polygonal obstacles~\cite{SoloveyHalperin2016}. These geometric results complement, but are not directly comparable with, the graph-based complexity results summarized in Table~\ref{tab:mapf-complexity}.

\section{Problem Formulation}
\label{sec:problem}

An Anonymous Multi-Agent Path Finding (AMAPF) instance is $\mathcal{I}=(G,S,D)$,
where $G=(V,E)$ is an undirected graph and
$S,D\subseteq V$ are sets of $N$ start and goal vertices, respectively.
The vertices within each set are distinct.
Initially, one agent occupies every vertex in $S$.

For each agent $a$, let $s_a\in S$ be its start and
$\pi_a:\mathbb{N}_0\rightarrow V$ its trajectory, where
$\pi_a(0)=s_a$.
At each timestep, an agent may wait or move to an adjacent vertex.
A set of trajectories is collision-free if no two agents occupy the
same vertex at the same time or traverse the same edge in opposite
directions during the same timestep.

A solution consists of a bijective assignment of the agents to the
goals in $D$ and collision-free trajectories such that each agent
eventually remains at its assigned goal forever. 
In particular, an agent starting at a goal need not be assigned to
that goal and may leave it.
For an agent $a$ assigned to goal $g_a\in D$, its completion time is
\(
    T_a
    :=
    \min\{t\geq0:
    \pi_a(t')=g_a,\ \forall t'\geq t\}.
\)
The sum-of-costs (SoC) is
\begin{equation}
    \SoC=\sum_a T_a.
    \label{eq:soc}
\end{equation}
The AMAPF-SoC problem is to find a solution minimizing
\eqref{eq:soc}.

\section{Flow Formulation and Loss of Integrality}
\label{sec:flow}

We first examine why the standard network-flow formulation for AMAPF
does not directly extend to goal-staying SoC.
Let $H$ be a time horizon large enough to contain an optimal solution. Such an $H$ can be chosen polynomial in the size of the input: a polynomial-makespan AMAPF solution exists~\cite{YuLaValle2013Flow}, yielding a polynomial upper bound on SoC and hence on the makespan of a minimum-SoC solution.

Define
\(
    \bar E
    =
    \{(u,v),(v,u):\{u,v\}\in E\}
    \cup
    \{(v,v):v\in V\},
\)
where $(v,v)$ represents waiting at $v$ for one timestep.
For every $(u,v)\in\bar E$ and $t=0,\ldots,H-1$, introduce a binary
flow variable
\(
    x_{uv}^t\in\{0,1\},
\)
which indicates whether an agent moves from $u$ at time $t$ to $v$ at
time $t+1$.
Because the agents are anonymous, the formulation uses a single
commodity and requires no agent index.

The initial configuration is given by
\begin{equation}
    \sum_{v:(u,v)\in\bar E} x_{uv}^{0}
    =
    \begin{cases}
        1, & u\in S,\\
        0, & u\notin S,
    \end{cases}
    \qquad \forall u\in V.
    \label{eq:flow-start}
\end{equation}
Flow conservation requires
\begin{equation}
    \sum_{u:(u,v)\in\bar E} x_{uv}^{t-1}
    =
    \sum_{w:(v,w)\in\bar E} x_{vw}^{t},
    \quad
    \forall v\in V,\;
    t=1,\ldots,H-1.
    \label{eq:flow-conservation}
\end{equation}
Vertex and edge conflicts are excluded by
\begin{align}
    \sum_{u:(u,v)\in\bar E} x_{uv}^{t-1}
    &\leq 1,
    &&
    \forall v\in V,\;
    t=1,\ldots,H,
    \label{eq:flow-vertex}\\
    x_{uv}^{t}+x_{vu}^{t}
    &\leq 1,
    &&
    \forall \{u,v\}\in E,\;
    t=0,\ldots,H-1.
    \label{eq:flow-edge}
\end{align}
Finally, every goal must be occupied at the end of the horizon:
\begin{equation}
    \sum_{u:(u,g)\in\bar E}x_{ug}^{H-1}=1,
    \qquad
    \forall g\in D.
    \label{eq:flow-goal}
\end{equation}

Constraints~\eqref{eq:flow-start}--\eqref{eq:flow-goal} are the
standard single-commodity time-expanded-flow constraints for AMAPF.
They suffice for feasibility and makespan optimization.
For SoC, however, the time at which a goal is first occupied is not
necessarily the completion time: an agent may enter a goal, leave it
later, and only subsequently settle at a goal permanently.

To capture this distinction, for every $g\in D$ and
$t=0,\ldots,H$, introduce a binary \emph{goal-settlement} variable
\(
    y_g^t\in\{0,1\},
\)
where $y_g^t=1$ if and only if the final occupant of $g$ remains at $g$ from time $t$
onward.
Since every goal is occupied at time $H$ and the solution can then be
extended by waiting forever,
\begin{equation}
    y_g^H=1,
    \qquad
    \forall g\in D.
    \label{eq:settle-terminal}
\end{equation}
For $t=0,\ldots,H-1$, permanent occupancy satisfies
\[
    y_g^t=x_{gg}^t\wedge y_g^{t+1},
\]
which can be linearized as
\begin{subequations}
\label{eq:settle}
\begin{align}
    y_g^t &\leq x_{gg}^t,
        \label{eq:settle-a}\\
    y_g^t &\leq y_g^{t+1},
        \label{eq:settle-b}\\
    y_g^t &\geq x_{gg}^t+y_g^{t+1}-1,
        \label{eq:settle-c}
\end{align}
\end{subequations}
for all $g\in D$ and $t=0,\ldots,H-1$.

If $T_g$ denotes the settlement time of goal $g$, i.e., the first
time from which its final occupant remains at $g$, then
\[
    y_g^t=
    \begin{cases}
        0, & t<T_g,\\
        1, & t\geq T_g.
    \end{cases}
\]
Since each completed agent permanently occupies a distinct goal,
$\{T_g:g\in D\}$ is exactly the multiset of agent completion times.
Therefore
\begin{equation}
    \SoC
    =
    \sum_{g\in D}T_g
    =
    \sum_{g\in D}\sum_{t=0}^{H-1}(1-y_g^t).
    \label{eq:flow-soc}
\end{equation}

The resulting formulation is
\begin{equation}
\begin{aligned}
    \min_{x,y}\quad&
        \sum_{g\in D}\sum_{t=0}^{H-1}(1-y_g^t)
        \\
    \text{s.t.}\quad&
        \eqref{eq:flow-start}-\eqref{eq:settle},\\
    &
        x_{uv}^t\in\{0,1\},
        \quad
        (u,v)\in\bar E,\;
        t=0,\ldots,H-1,\\
    &
        y_g^t\in\{0,1\},
        \quad
        g\in D,\;
        t=0,\ldots,H.
\end{aligned}
\tag{IP-SoC}
\label{eq:ip-soc}
\end{equation}

\begin{proposition}
\label{prop:ip-correct}
For any horizon $H$ at least the makespan of a minimum-SoC AMAPF
solution, an optimal solution of \eqref{eq:ip-soc} corresponds to a
minimum-SoC AMAPF solution, with the same SoC.
\end{proposition}

\begin{proof}
Constraints~\eqref{eq:flow-start}--\eqref{eq:flow-goal} describe a
collision-free anonymous flow from the start configuration at time
$0$ to the goal configuration at time $H$.
For a fixed goal $g$, Constraints~\eqref{eq:settle-terminal} and
\eqref{eq:settle} imply recursively that $y_g^t=1$ exactly when all
wait arcs
\(
    (g,t)\rightarrow(g,t+1),\ldots,
    (g,H-1)\rightarrow(g,H)
\)
are occupied.
Thus $T_g$ is precisely the time at which the final occupant of $g$
arrives and remains there through the end of the horizon.
Extending the solution by waiting after time $H$ makes $T_g$ its
completion time.
Equation~\eqref{eq:flow-soc} therefore equals the sum of all agent
completion times, and minimizing it is equivalent to minimizing SoC.
\end{proof}

The distinction from disappearing AMAPF is now explicit.
If agents disappear after reaching goals, a unit of flow may simply
leave the time-expanded network through a goal at time $t$, with cost
$t$, yielding a standard minimum-cost maximum-flow problem (with additional gadget constructions to exclude vertex and edge conflicts~\cite{YuLaValle2013Flow,AliYakovlev2024}).
With goal staying, completing at a goal instead reserves that vertex
for every subsequent timestep.
The settlement variables encode exactly this persistent reservation.

Although the underlying single-commodity flow constraints are
integral, adding this settlement condition destroys the corresponding
integrality property.
The natural LP relaxation of \eqref{eq:ip-soc} can fractionally settle
a goal, as we show next.

\subsection{Integrality Gap of the LP Relaxation}
\label{subsec:integrality-gap}

The LP relaxation of \eqref{eq:ip-soc}, obtained by replacing the
binary constraints with
\(
    0\leq x_{uv}^t,y_g^t\leq1,
\)
is not integral.
Consider the six-vertex tree in
Figure~\ref{fig:lp-gap-example}, with $S=\{a,g_L,g_R,b\}$, $D=\{g_L,c,g_R,h\}$, and horizon $H=2$.

\begin{figure}[t]
\centering
\begin{tikzpicture}[
    x=1cm,
    y=0.92cm,
    font=\small,
    vertex/.style={
        circle, draw, minimum size=6.5mm, inner sep=0pt
    },
    start/.style={
        vertex, fill=gray!20
    },
    goal/.style={
        vertex, double, double distance=0.8pt
    },
    both/.style={
        vertex, double, double distance=0.8pt, fill=gray!20
    },
]
\node[start] (a)  at (0,0) {$a$};
\node[both]  (gL) at (1.8,0) {$g_L$};
\node[goal]  (c)  at (3.6,0) {$c$};
\node[both]  (gR) at (5.4,0) {$g_R$};
\node[start] (b)  at (7.2,0) {$b$};
\node[goal]  (h)  at (3.6,-1.8) {$h$};

\draw (a)--(gL)--(c)--(gR)--(b);
\draw (c)--(h);
\end{tikzpicture}
\caption{
A six-vertex AMAPF instance.
Gray-filled vertices are starts, double-circled vertices are goals,
and $g_L,g_R$ are both starts and goals.
}
\label{fig:lp-gap-example}
\end{figure}
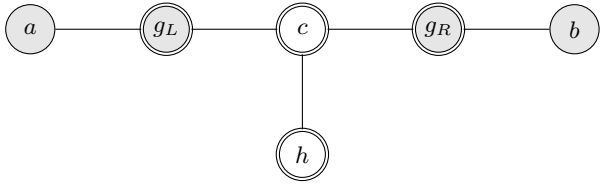

\paragraph{Integer optimum.}
A solution of SoC $7$ is
{\renewcommand{\arraystretch}{0.6}
\[
\begin{array}{c|ccc}
 & t=0 & t=1 & t=2 \\ \hline
A & a   & a   & g_L \\
L & g_L & g_L & c   \\
R & g_R & c   & h   \\
B & b   & g_R & g_R
\end{array} \text{(all agents waiting for all $t>2$)},
\]}%
with completion times $2,2,2,1$.
To see that no integer solution can do better, first observe that
\(
    T_{g_L}+T_{g_R}\geq3.
\)
Both goals must have completion time at least one, since the agents at
$a$ and $b$ can enter the rest of the graph only through $g_L$ and
$g_R$.  If both were permanently occupied from time $1$, the agents
initially at $g_L$ and $g_R$ would both have to move to $c$ at time
$0\to1$, causing a collision.
Moreover,
\(
    T_c+T_h\geq4.
\)
Indeed, $T_h\geq2$.  If $T_h=2$, an agent must occupy $c$ at time $1$
and move to $h$, so $c$ cannot yet be permanently occupied and
$T_c\geq2$.  If $T_h\geq3$, then $T_c\geq1$ suffices to give the same
bound.  Hence
\(
    \SoC
    =
    T_{g_L}+T_{g_R}+T_c+T_h
    \geq7,
\)
and therefore
\(
    \mathrm{OPT}_{\mathrm{IP}}=7.
\)

\paragraph{LP optimum.}
The relaxation can attain value $6$ by splitting each of the four
initial units of flow symmetrically.
During $0\to1$, half of the flow at $a$ waits and half moves to $g_L$,
while half of the flow initially at $g_L$ waits and half moves to $c$.
The same construction is applied symmetrically on the right.
Thus, at time $1$, $\operatorname{occ}(a)=\operatorname{occ}(b)=\tfrac12$, $
    \operatorname{occ}(g_L)
    =
    \operatorname{occ}(c)
    =
    \operatorname{occ}(g_R)
    =1$.
During $1\to2$, the unit of flow at $c$ moves to $h$.
At each of $g_L$ and $g_R$, half of the flow waits, while the other
half moves to $c$; simultaneously, the remaining half-unit at each
outer vertex moves into the corresponding goal.
All four goals therefore contain one unit of flow at time $2$.
For the settlement variables, set $
    y_{g_L}^0=y_{g_L}^1
    =
    y_{g_R}^0=y_{g_R}^1
    =\tfrac12$,
and
$y_c^0=y_c^1=y_h^0=y_h^1=0$.
Together with $y_g^2=1$ for every goal, these values satisfy
\eqref{eq:settle}.
The resulting objective value is
$1+1+2+2=6$.
For completeness, $6$ is also a lower bound on the LP objective.
Since $h$ can be reached only through $c$, satisfying the terminal
goal constraint at $h$ forces the entire unit of flow at $c$ at time
$1$ to move to $h$.
Consequently, $y_c^0=y_c^1=y_h^0=y_h^1=0$.
Furthermore, one unit of flow must enter $c$ from $g_L$ and $g_R$
during $0\to1$, which implies $y_{g_L}^0+y_{g_R}^0\leq1$.
Similarly, since one unit must enter $c$ from $g_L$ and $g_R$ during
$1\to2$ to occupy $c$ at time $2$, $y_{g_L}^1+y_{g_R}^1\leq1$.
Hence $\sum_{g\in D}\sum_{t=0}^{1}y_g^t\leq2$, and the LP objective is at least $8-2=6$.

We have therefore proved the following.

\begin{proposition}
\label{prop:integrality-gap}
The LP relaxation of \eqref{eq:ip-soc} is not integral.
For the instance in Figure~\ref{fig:lp-gap-example}, $\mathrm{OPT}_{\mathrm{LP}}=6$ and $\mathrm{OPT}_{\mathrm{IP}}=7$, giving an integrality gap of $7/6$.
\end{proposition}

The relaxation succeeds by \emph{fractionally settling} $g_L$ and
$g_R$: part of a unit of flow permanently remains at each goal while
the remaining part continues through the graph.
An integral agent cannot make such a split.
This loss of integrality suggests that permanent goal occupancy
introduces genuine combinatorial choices beyond standard
single-commodity network flow.

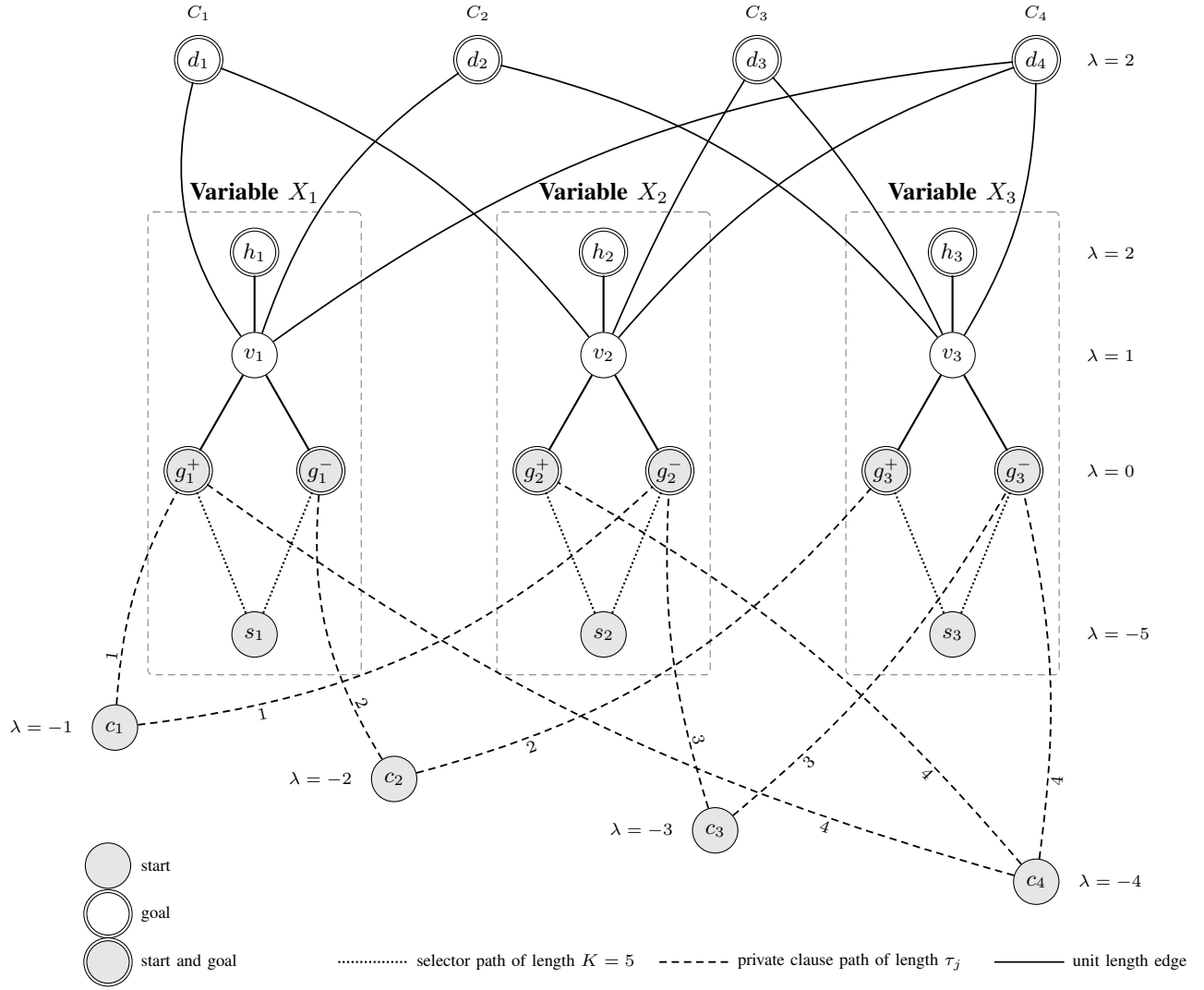
\begin{figure*}[t]
\centering
\begin{tikzpicture}[
    x=1cm,
    y=0.92cm,
    font=\small,
    vertex/.style={
        circle, draw, minimum size=6.5mm, inner sep=0pt
    },
    start/.style={
        vertex, fill=gray!20
    },
    goal/.style={
        vertex, double, double distance=0.8pt
    },
    both/.style={
        vertex, double, double distance=0.8pt, fill=gray!20
    },
    internal/.style={
        vertex
    },
    selectorpath/.style={
        densely dotted, line width=0.8pt
    },
    clausepath/.style={
        densely dashed, line width=0.7pt
    },
    incidence/.style={
        line width=0.65pt
    },
    local/.style={
        line width=0.8pt
    },
    gadgetbox/.style={
        draw, gray, rounded corners=2pt, densely dashed,
        inner sep=7pt
    },
    pathlabel/.style={
        fill=white, inner sep=1.2pt, font=\scriptsize
    }
]


\node[start]    (s1)  at (0,-4.8) {$s_1$};
\node[both]     (g1p) at (-0.95,-2.25) {$g_1^+$};
\node[both]     (g1m) at ( 0.95,-2.25) {$g_1^-$};
\node[internal] (v1)  at (0,-0.45) {$v_1$};
\node[goal]     (h1)  at (0,1.15) {$h_1$};

\draw[selectorpath]
  (s1) to[bend left=0]
  node[pathlabel,sloped,above] {} (g1p);
\draw[selectorpath]
  (s1) to[bend right=0]
  node[pathlabel,sloped,above] {} (g1m);

\draw[local] (g1p)--(v1);
\draw[local] (g1m)--(v1);
\draw[local] (v1)--(h1);

\node[gadgetbox,
      fit=(s1)(g1p)(g1m)(v1)(h1),
      label={[font=\bfseries]above:Variable $X_1$}] {};


\node[start]    (s2)  at (5,-4.8) {$s_2$};
\node[both]     (g2p) at (4.05,-2.25) {$g_2^+$};
\node[both]     (g2m) at (5.95,-2.25) {$g_2^-$};
\node[internal] (v2)  at (5,-0.45) {$v_2$};
\node[goal]     (h2)  at (5,1.15) {$h_2$};

\draw[selectorpath]
  (s2) to[bend left=0]
  node[pathlabel,sloped,above] {} (g2p);
\draw[selectorpath]
  (s2) to[bend right=0]
  node[pathlabel,sloped,above] {} (g2m);

\draw[local] (g2p)--(v2);
\draw[local] (g2m)--(v2);
\draw[local] (v2)--(h2);

\node[gadgetbox,
      fit=(s2)(g2p)(g2m)(v2)(h2),
      label={[font=\bfseries]above:Variable $X_2$}] {};


\node[start]    (s3)  at (10,-4.8) {$s_3$};
\node[both]     (g3p) at (9.05,-2.25) {$g_3^+$};
\node[both]     (g3m) at (10.95,-2.25) {$g_3^-$};
\node[internal] (v3)  at (10,-0.45) {$v_3$};
\node[goal]     (h3)  at (10,1.15) {$h_3$};

\draw[selectorpath]
  (s3) to[bend left=0]
  node[pathlabel,sloped,above] {} (g3p);
\draw[selectorpath]
  (s3) to[bend right=0]
  node[pathlabel,sloped,above] {} (g3m);

\draw[local] (g3p)--(v3);
\draw[local] (g3m)--(v3);
\draw[local] (v3)--(h3);

\node[gadgetbox,
      fit=(s3)(g3p)(g3m)(v3)(h3),
      label={[font=\bfseries]above:Variable $X_3$}] {};



%

\node[goal] (d1) at (-0.8,4.15) {$d_1$};
\node[goal] (d2) at ( 3.2,4.15) {$d_2$};
\node[goal] (d3) at ( 7.2,4.15) {$d_3$};
\node[goal] (d4) at (11.2,4.15) {$d_4$};

\node[font=\scriptsize,anchor=south]
    at ($(d1.north)+(0,0.12)$) {$C_1$};
\node[font=\scriptsize,anchor=south]
    at ($(d2.north)+(0,0.12)$) {$C_2$};
\node[font=\scriptsize,anchor=south]
    at ($(d3.north)+(0,0.12)$) {$C_3$};
\node[font=\scriptsize,anchor=south]
    at ($(d4.north)+(0,0.12)$) {$C_4$};

%

\draw[incidence] (v1) to[bend left=25] (d1);
\draw[incidence] (v2) to[bend right=15] (d1);

\draw[incidence] (v1) to[bend left=18] (d2);
\draw[incidence] (v3) to[bend right=18] (d2);

\draw[incidence] (v2) to[bend left=4] (d3);
\draw[incidence] (v3) to[bend right=8] (d3);

\draw[incidence] (v1) to[bend left=15] (d4);
\draw[incidence] (v2) to[bend left=15] (d4);
\draw[incidence] (v3) to[bend right=15] (d4);

%

\node[start] (c1) at (-2.0,-6.25) {$c_1$};
\node[start] (c2) at ( 2.0,-7.05) {$c_2$};
\node[start] (c3) at ( 6.6,-7.85) {$c_3$};
\node[start] (c4) at (11.2,-8.65) {$c_4$};

\node[font=\scriptsize,anchor=east]
  at ($(c1.west)+(-0.15,0)$) {$\lambda=-1$};
\node[font=\scriptsize,anchor=east]
  at ($(c2.west)+(-0.15,0)$) {$\lambda=-2$};
\node[font=\scriptsize,anchor=east]
  at ($(c3.west)+(-0.15,0)$) {$\lambda=-3$};
\node[font=\scriptsize,anchor=west]
  at ($(c4.east)+(0.15,0)$) {$\lambda=-4$};

%


\draw[clausepath]
  (c1) to[bend left=12]
  node[pathlabel,sloped,above,pos=0.2] {$1$} (g1p);

\draw[clausepath]
  (c1) to[bend right=18]
  node[pathlabel,sloped,below,pos=0.2] {$1$} (g2m);


\draw[clausepath]
  (c2) to[bend left=18]
  node[pathlabel,sloped,above,pos=0.2] {$2$} (g1m);

\draw[clausepath]
  (c2) to[bend right=17]
  node[pathlabel,sloped,below,pos=0.2] {$2$} (g3p);


\draw[clausepath]
  (c3) to[bend left=10]
  node[pathlabel,sloped,above,pos=0.2] {$3$} (g2m);

\draw[clausepath]
  (c3) to[bend right=10]
  node[pathlabel,sloped,below,pos=0.2] {$3$} (g3m);


\draw[clausepath]
  (c4) to[bend left=10]
  node[pathlabel,sloped,below,pos=0.2] {$4$} (g1p);

\draw[clausepath]
  (c4) to[bend right=10]
  node[pathlabel,sloped,below,pos=0.2] {$4$} (g2p);

\draw[clausepath]
  (c4) to[bend right=10]
  node[pathlabel,sloped,below,pos=0.2] {$4$} (g3m);


\node[font=\scriptsize,anchor=west]
    at (11.8,-4.8) {$\lambda=-5$};

\node[font=\scriptsize,anchor=west]
    at (11.8,-2.25) {$\lambda=0$};

\node[font=\scriptsize,anchor=west]
    at (11.8,-0.45) {$\lambda=1$};

\node[font=\scriptsize,anchor=west]
    at (11.8,1.15) {$\lambda=2$};

\node[font=\scriptsize,anchor=west]
    at (11.8,4.15) {$\lambda=2$};




\node[start] (legS) at (-2.1,-8.4) {};
\node[anchor=west,font=\scriptsize] at (-1.75,-8.4) {start};

\node[goal] (legG) at (-2.1,-9.15) {};
\node[anchor=west,font=\scriptsize] at (-1.75,-9.15) {goal};

\node[both] (legB) at (-2.1,-9.9) {};
\node[anchor=west,font=\scriptsize] at (-1.75,-9.9) {start and goal};

\draw[selectorpath] (1.2,-9.9)--(2.2,-9.9);
\node[anchor=west,font=\scriptsize] at (2.2,-9.9) {selector path of length $K=5$};

\draw[clausepath] (5.8,-9.9)--(6.8,-9.9);
\node[anchor=west,font=\scriptsize] at (6.8,-9.9) {private clause path of length $\tau_j$};

\draw[incidence] (10.6,-9.9)--(11.6,-9.9);
\node[anchor=west,font=\scriptsize] at (11.6,-9.9) {unit length edge};

\end{tikzpicture}

\caption{
Example of the reduction for
$\Phi=(X_1\vee\neg X_2)\wedge(\neg X_1\vee X_3)
\wedge(\neg X_2\vee\neg X_3)\wedge
(X_1\vee X_2\vee\neg X_3)$.
Here $m=4$, so $K=5$ and
$(\tau_1,\tau_2,\tau_3,\tau_4)=(1,2,3,4)$.
The clause starts $c_j$, of potential $-\tau_j$, are shown below the
variable gadgets, whereas all clause goals $d_j$ have potential $2$
and are shown above them.
A dotted connection labeled $5$ represents an internally disjoint
selector path of exactly $K$ unit edges, and a dashed connection
labeled $\tau_j$ represents an internally private clause path
$Q_{j,\ell}$ of exactly $\tau_j$ unit edges.
All other displayed connections are individual edges.
The literal gates $g_i^+$ and $g_i^-$ are simultaneously starts and
goals, and every path from either gate toward a positive-potential
goal passes through the common bottleneck $v_i$.
}
\label{fig:amapf-reduction-example}
\end{figure*}

\section{NP-Hardness of Optimal AMAPF-SoC}
\label{sec:hardness}

We now show that the loss of integrality identified above reflects an
inherent computational difficulty.

\begin{theorem}
\label{thm:hardness}
Computing a minimum-SoC solution for AMAPF is NP-hard.
\end{theorem}

We reduce from 3-SAT.
Let
\(
    \Phi=\bigwedge_{j=1}^{m} C_j
\)
be a 3-CNF formula over variables $X_1,\ldots,X_n$, where each clause $C_j$ contains at most three literals.
We may assume that each clause contains distinct literals and no
complementary pair.

\subsection{Potential Lower Bound}

The reduction uses a potential function to characterize solutions that
attain a tight SoC lower bound.
Let
\(
    \lambda:V\rightarrow\mathbb{Z}
\)
satisfy
\begin{equation}
    |\lambda(u)-\lambda(v)|\leq1,
    \qquad
    \forall\{u,v\}\in E.
    \label{eq:potential-lipschitz}
\end{equation}

\begin{lemma}
\label{lem:potential}
Every feasible AMAPF solution satisfies
\begin{equation}
    \SoC
    \geq
    B_\lambda
    :=
    \sum_{g\in D}\lambda(g)
    -
    \sum_{s\in S}\lambda(s).
    \label{eq:potential-bound}
\end{equation}
If equality holds, then every agent, before completion, moves
at every timestep along an edge that increases $\lambda$ by exactly
one.
\end{lemma}

\begin{proof}
Consider an agent $a$ starting at $s$ and completing at $g$ at time
$T_a$.
By \eqref{eq:potential-lipschitz},
\[
    \lambda(g)-\lambda(s)
    =
    \sum_{t=0}^{T_a-1}
    \bigl(
        \lambda(\pi_a(t+1))-\lambda(\pi_a(t))
    \bigr)
    \leq T_a.
\]
Summing over all agents and using the bijection between agents and
goals gives \eqref{eq:potential-bound}.
Equality can hold only if every individual inequality is tight, which
requires every move before completion to increase the potential by
exactly one.
\end{proof}

\subsection{Reduction}

Set $\tau_j:=j$, $K:=m+1$.
Hence $1\leq\tau_j<K$,
and all $\tau_j$ are distinct.
Unless stated otherwise, all paths introduced below are internally
vertex-disjoint.

\paragraph{Variable gadget.}
For each variable $X_i$, introduce two \emph{literal gates}
$g_i^+$ and $g_i^-$, a shared bottleneck $v_i$, and an auxiliary goal
$h_i$, with
\[
    \lambda(g_i^+)=\lambda(g_i^-)=0,
    \qquad
    \lambda(v_i)=1,
    \qquad
    \lambda(h_i)=2.
\]
Add
\[
    (g_i^+,v_i),\qquad
    (g_i^-,v_i),\qquad
    (v_i,h_i).
\]
Both literal gates are simultaneously starts and goals.

Introduce an additional selector start $s_i$ with
\[
    \lambda(s_i)=-K,
\]
and connect it to $g_i^+$ and $g_i^-$ by two internally disjoint paths
$P_i^+$ and $P_i^-$, each of length $K$.
The potential increases by one along every edge of either path.

Thus the variable gadget contributes starts
\[
    g_i^+,\quad g_i^-,\quad s_i
\]
and goals
\[
    g_i^+,\quad g_i^-,\quad h_i.
\]

The shared bottleneck is the key feature.
In any solution attaining the potential lower bound, an agent leaving
either literal gate must move to $v_i$ at time $1$.
Hence the two literal-gate agents cannot both leave.

\paragraph{Clause gadget.}
For each clause $C_j$, introduce a clause start $c_j$ and a clause goal
$d_j$ with
\[
    \lambda(c_j)=-\tau_j,
    \qquad
    \lambda(d_j)=2.
\]
For a literal $\ell$, define
\[
    g(\ell)=
    \begin{cases}
        g_i^+, & \ell=X_i,\\
        g_i^-, & \ell=\neg X_i.
    \end{cases}
\]
For every occurrence $\ell\in C_j$, add an internally private path
$Q_{j,\ell}$ of length $\tau_j$ from $c_j$ to $g(\ell)$, with
potential increasing by one along every edge.

Finally, for every variable $X_i$ occurring in $C_j$, add the edge
\[
    (v_i,d_j).
\]
Thus a clause agent choosing literal $\ell\in C_j$, whose variable is
$X_i$, can follow
\[
    c_j
    \xrightarrow[\tau_j]{Q_{j,\ell}}
    g(\ell)
    \to v_i
    \to d_j,
    \label{eq:clause-route}
\]
using $\tau_j+2$ steps.

An example of the construction is shown in
Figure~\ref{fig:amapf-reduction-example}.
The polarity of an occurrence is encoded by the literal gate through
which the clause agent enters; the outgoing edge from $v_i$ depends
only on the variable.

If a connected graph is desired, add
\[
    (s_i,s_{i+1}),
    \qquad i=1,\ldots,n-1.
\]
Both endpoints have potential $-K$, so no solution attaining
\eqref{eq:potential-bound} can use these edges.

The complete start and goal sets are
\begin{align}
S
&=
\{g_i^+,g_i^-,s_i:1\leq i\leq n\}
\cup
\{c_j:1\leq j\leq m\},\\
D
&=
\{g_i^+,g_i^-,h_i:1\leq i\leq n\}
\cup
\{d_j:1\leq j\leq m\}.
\end{align}
The instance therefore contains $3n+m$ agents.
Since $K=m+1$, the construction has polynomial size.

\subsection{SoC Threshold}

The total goal potential is
\[
    \sum_{g\in D}\lambda(g)=2n+2m,
\]
whereas the total start potential is
\[
    \sum_{s\in S}\lambda(s)
    =
    -nK-\sum_{j=1}^{m}\tau_j.
\]
Lemma~\ref{lem:potential} therefore gives
\begin{equation}
    \SoC\geq
    B
    :=
    n(K+2)
    +
    \sum_{j=1}^{m}(\tau_j+2).
    \label{eq:threshold}
\end{equation}
We show that $\Phi$ is satisfiable if and only if the constructed
instance admits a solution of SoC $B$.

\subsection{Correctness}

\begin{lemma}
\label{lem:complete}
If $\Phi$ is satisfiable, then the constructed AMAPF instance has a
collision-free solution of SoC $B$.
\end{lemma}

\begin{proof}
Fix a satisfying assignment.
For each $X_i$, let $g_i^{\sigma_i}$ denote the gate corresponding to
its true literal and $g_i^{-\sigma_i}$ the other gate.
Keep the agent at $g_i^{-\sigma_i}$ there permanently.
Route the agent at $g_i^{\sigma_i}$ as $g_i^{\sigma_i}\to v_i\to h_i$,
and route the selector at $s_i$ along $P_i^{\sigma_i}$ to the vacated
gate, which it reaches at time $K$.
For every clause $C_j$, choose a true literal $\ell_j\in C_j$ and
route its agent as \(c_j
    \xrightarrow[\tau_j]{Q_{j,\ell_j}}
    g(\ell_j)
    \to v_i
    \to d_j\),
where $X_i$ is the variable of $\ell_j$.
The clause agent reaches its goal at time $\tau_j+2$.
These trajectories are collision-free.
The false gate is occupied permanently, whereas clauses use only true
gates.
The moving literal-gate agent vacates its gate at time $0$ and leaves
$v_i$ at time $2$, so a clause reaching the gate as early as time $1$
can follow immediately behind it.
Likewise, $\tau_j<K$ ensures that every clause leaves its selected gate
before the selector arrives at time $K$.
Finally, the $\tau_j$ are pairwise distinct, so clauses using the same
literal reach its gate at different times and pipeline through
$g(\ell)\to v_i$ without collision.
All other relevant paths are internally private.

The completion times contributed by each variable are 0, 2, and $K$,
and clause $C_j$ contributes $\tau_j+2$.
Hence the total SoC is exactly \(n(K+2)+\sum_{j=1}^{m}(\tau_j+2)=B\).
\end{proof}

\begin{lemma}
\label{lem:sound}
If the constructed AMAPF instance has a solution of SoC at most $B$,
then $\Phi$ is satisfiable.
\end{lemma}

\begin{proof}
Suppose that $\SoC\leq B$.
By Lemma~\ref{lem:potential} and \eqref{eq:threshold},
$\SoC=B$.
Thus every agent must move exclusively along edges that increase
$\lambda$ by one until it completes its path.

Consider a variable $X_i$.
An agent initially at $g_i^+$ or $g_i^-$ starts at potential zero.
If it completes at a level-zero goal, tightness requires completion at
time zero, so it must stay at its own starting gate.
Otherwise, its only potential-increasing first move is to $v_i$.
Consequently, if both literal-gate agents left, both would occupy
$v_i$ at time $1$, which is impossible.
Thus at most one leaves.

On the other hand, both cannot stay.
The selector at $s_i$ must move along one of $P_i^+$ and $P_i^-$ and
reaches the corresponding literal gate after $K$ steps.
If both gates were occupied permanently, no tight trajectory would be
available to the selector.
Therefore exactly one literal-gate agent leaves.

Define
\begin{equation}
    X_i=\mathsf{True}
    \quad\Longleftrightarrow\quad
    \text{the agent at $g_i^+$ leaves}.
    \label{eq:truth-assignment}
\end{equation}
Hence the vacated gate is exactly the gate corresponding to the true
literal of $X_i$.

We next show that the selector must complete at this vacated gate.
The opposite gate is permanently occupied, so the selector must follow
the path to the vacated gate and reaches it at time $K$.
If it did not complete there, some other agent would have to occupy
that level-zero goal permanently.
The original gate agent cannot return because every tight move
strictly increases potential, and no selector from another variable
can reach it.
A clause agent can reach the gate only at some time
$\tau_j<K$; if it completed there, goal staying would block the
selector from entering at time $K$.
Thus the selector completes at the vacated gate.

It follows that all $2n$ level-zero goals are occupied by the
stationary literal-gate agents and selectors.
The remaining agents, including every clause agent, must therefore
complete at level-two goals.

Consider the clause agent starting at $c_j$.
Since it starts at potential $-\tau_j$ and completes at potential
$2$, tightness requires exactly $\tau_j+2$ moves.
Its only potential-increasing choices from $c_j$ are the occurrence
paths $Q_{j,\ell}$ for $\ell\in C_j$.
It must therefore reach some gate $g(\ell)$ at time $\tau_j$.

If $\ell$ were false under \eqref{eq:truth-assignment}, then
$g(\ell)$ would be the gate occupied permanently by its original
literal-gate agent, so the clause agent could not enter it.
Hence $\ell$ must be true.

Every clause therefore contains a true literal, and the extracted
assignment satisfies $\Phi$.
\end{proof}

\begin{proof}[Proof of Theorem~\ref{thm:hardness}]
The reduction has polynomial size.
By Lemma~\ref{lem:complete}, a satisfiable formula yields an AMAPF
solution of SoC $B$.
By Lemma~\ref{lem:sound}, any AMAPF solution of SoC at most $B$ yields
a satisfying assignment.
Thus deciding whether the optimal SoC is at most $B$ is NP-hard, and so is computing a minimum-SoC AMAPF solution.
\end{proof}

The reduction relies only on the permanent occupancy of the
level-zero literal goals.
The shared bottleneck $v_i$ prevents both literal gates from being
vacated, while the selector prevents both from remaining occupied.
Together they force a binary choice.
Anonymous reassignment among the level-two goals is immaterial: the
soundness argument requires only that each clause agent pass through a
gate belonging to its own clause before reaching level two.

\section{Discussion and Conclusion}
\label{sec:conclusion}

We showed that minimizing SoC in standard goal-staying AMAPF is NP-hard,
closing the complexity gap between standard AMAPF and its
polynomial-time flow-based variants.
Our analysis identifies permanent goal occupancy as the key source of
difficulty: adding goal-settlement constraints to the standard
time-expanded flow formulation destroys its integrality. Together with the polynomial-time solvability of the disappearing variant, our result establishes a sharp complexity boundary determined solely by whether completed agents remain at their goals.

Our reduction is also flexible. Its main ingredients are a tight potential lower bound, a shared bottleneck enforcing a binary choice, and permanently occupied literal goals communicating this choice to clause agents. The clause paths are otherwise private, and their different lengths serve only to separate clause agents in time. Thus the construction does not require bounded variable occurrence and can potentially be adapted to other NP-complete SAT variants.
Such adaptations may yield stronger results. With bounded-occurrence SAT, each variable has only constantly many clause connections, and the polynomial path lengths in our construction can also be replaced by constant ones. Specifically, the values $\tau_j=1,\ldots,m$ are used only to separate clauses sharing a variable, while $K=m+1$ places the selector after all such clauses. Under bounded occurrence, a constant number of time offsets suffices, for example by coloring the conflict graph of clauses sharing a variable. Hence both $\tau_j$ and $K$ can be bounded by constants, giving constant-length clause and selector paths and potentially enabling bounded-degree reductions. Planar SAT variants may similarly lead to hardness on planar graphs, and suitable embedding gadgets may further extend the result to 2D grids. Finally, gap-preserving reductions from suitable SAT or Max-SAT variants may yield inapproximability results for goal-staying AMAPF.

\bibliographystyle{IEEEtran}
\bibliography{references}

\end{document}